\documentclass[pra,superscriptaddress]{revtex4-1}

\usepackage{mathtools}
\usepackage{amsthm}
\usepackage{amsmath}
\usepackage{amsfonts}
\usepackage{amssymb}
\usepackage{bm}
\usepackage{braket}
\usepackage{mathrsfs}
\usepackage{bbold}
\usepackage{verbatim}
\usepackage{color}
\definecolor{mygreen}{RGB}{28,172,0} 
\definecolor{mylilas}{RGB}{170,55,241}
\usepackage{listings}
\usepackage{bbm}

\newcommand{\defeq}{\vcentcolon=}
\newcommand{\tr}[1]{ \mathrm{tr}\left(#1\right) }

\newcommand{\norm}[1]{\left\lVert #1\right\rVert}

\newtheorem{Theorem}{Theorem}
\newtheorem*{Theorem*}{Theorem}
\newtheorem{Lemma}{Lemma}
\newtheorem*{Lemma*}{Lemma}
\newtheorem{Definition}{Definition}
\newtheorem*{Definition*}{Definition}

\begin{document}
%\begin{titlepage}
   \begin{center}
       \vspace*{1cm}

       \large{\textbf{Multipartite Entanglement Detection
via Correlation Minor Norm}}

       \vspace{0.5cm}

       %\vspace{1cm}

       \small{\textbf{Rain Lenny$^1$, Amit Te’eni$^1$, Bar Y. Peled$^1$, Avishy Carmi$^2$ and Eliahu Cohen$^1$}}
       
        \vspace{0.5cm}
        
    \small{$^{1}$Faculty of Engineering and the Institute of Nanotechnology and Advanced Materials, Bar Ilan University, Ramat Gan, 5290002, Israel\\
$^{2}$Center for Quantum Information Science and Technology and Faculty of Engineering Sciences, Ben-Gurion University of the Negev, Beersheba, 8410501, Israel}
       
   \vspace{1cm}
   
    \textbf{Abstract}
    
Entanglement is a uniquely quantum resource giving rise to many quantum technologies. It is therefore important to detect and characterize entangled states, but this is known to be a challenging task, especially for multipartite mixed states. The correlation minor norm (CMN) was recently suggested as a bipartite entanglement detector employing bounds on the quantum correlation matrix. In this paper we explore generalizations of the CMN to multipartite systems based on matricizations of the correlation tensor. It is shown that the CMN is able to detect and differentiate classes of multipartite entangled states. We further analyze the correlations within the reduced density matrices and show their significance for entanglement detection. Finally, we employ matricizations of the correlation tensor for introducing a measure of global quantum discord.
   \end{center}
       %\vfill

       %\vspace{0.8cm}

      %30/01/2023

%\end{titlepage}

\section{Introduction}

\begin{comment}

The study of quantum mechanics in the last four decades led to numerous applications in various fields: In computation, the Deutsch-Jozsa algorithm \cite{deutsch1985quantum,deutsch1992rapid} was one of the first quantum algorithms to be more efficient than its classical counterpart. Shortly after, Shor's algorithm for finding prime factors \cite{shor1994algorithms} and Grover's search algorithm \cite{grover1996fast} followed, which again showed the potential strength of quantum computation. In communication, superdense coding \cite{bennett1992communication} allowed to transmit two bits using one qubit, and quantum teleportation \cite{bennett1993teleporting} allowed to transmit a quantum state using only classical communications. In cryptography, the BB84 \cite{bennett2020quantum} and the E91 \cite{Ekert1991QuantumCB} protocols presented methods for eavesdropping-immune key distribution. In sensing, entanglement can be used for ghost imaging \cite{d2005resolution,zerom2011entangled}, and 
squeezing can be used for enhanced sensing, which was applied in the Laser Interferometer Gravitational-Wave Observatory (LIGO) for example \cite{aasi2013enhanced,andersen2013squeezing}.
\end{comment}

In applications of quantum mechanics, entangled particles are often used to create correlations stronger than classically possible. Theoretically, determining if entanglement exists in an arbitrary system represented by a density operator, is a notoriously difficult problem. Practically, determining whether entanglement is present in an experimental setup can prove to be even more challenging, considering the needed precision as well as the existence of noise. Moreover, trying to characterize and use entanglement in many-body quantum systems raises several questions: how strong is the entanglement after it is produced? Is the state genuinely entangled (or just bi-separable for instance)? Can one easily measure the strength of entanglement? Are their additional helpful quantities characterizing the nonlocality of the system?

In trying to answer these questions, many entanglement detection schemes were suggested. Notable examples include the Peres-Horodecki criterion \cite{peres1996separability,HORODECKI19961}, the computable cross-norm or realignment (CCNR) criterion \cite{chen2002matrix,rudolph2005further}, entanglement witnesses    \cite{HORODECKI19961} and de Vicente's criterion  \cite{vicente2007separability} (for more information on the basics of entanglement detection, see Ref. \cite{guhne2009entanglement}). Among them, the correlation minor norm \cite{peled2021correlation} (CMN), based on a previous more general construction of the same authors \cite{carmi2019relativistic}, was presented as a bipartite entanglement detection scheme using the correlation matrix, a tool heavily used in signal processing which encodes important information about entanglement in quantum systems \cite{badziag2008experimentally,de2011multipartite}.

Considering that entanglement detection in mixed bipartite states can present a challenge, detection of entanglement in mixed multipartite systems is even more arduous. It is the goal of the present work to suggest and analyze a multipartite entanglement detection scheme based on the CMN. In our work we construct the correlation tensor for an arbitrary finite-dimensional mixed state and explore its properties in order to create an entanglement detection criterion. Different types of high order entanglement are discussed, as well as higher order quantum discord.  

\section{Preliminaries}
\subsection{construction of the correlation tensor}
Let remote parties $A,B,...,N$ share a quantum system in the following Hilbert space: $\mathcal{H}_A \otimes \mathcal{H}_B \otimes...\otimes \mathcal{H}_N$. Denote $d_A \defeq \dim(\mathcal{H}_A), d_B \defeq \dim(\mathcal{H}_B)$ and $d_N \defeq \dim(\mathcal{H}_N)$. For every party there always exists an orthonormal basis of $d_A \times d_A$ Hermitian operators: $\{A_i\}_{i=1}^{d_A^2}$, which sustains the following normalization: $\tr{A_iA_j} = \delta_{ij}$ (same for the other parties).

The (cross-)correlation tensor of the system, denoted by $\mathcal{C}$, is defined by:
\begin{equation}
    \mathcal{C}^{i,j..,n} \defeq \braket{A_i \otimes B_j \otimes...\otimes N_n} = \tr{\rho (A_i \otimes B_j \otimes...\otimes N_n)},
\end{equation}
where $\rho$ is the density matrix shared by all the systems. Each entry in the tensor describes the cross-correlation between the corresponding observables. Note that the correlation tensor is real.

\subsection{Quantum states and entanglement}

 For our discussion, entanglement is the lack of separability. A quantum state $\rho$ is fully-separable if and only if:
 \begin{equation}
  \rho=\sum_{i}{p_i{\ \rho}_i^{(A)}\otimes\rho_i^{(B)}\otimes\ldots{\otimes\rho}_i^{(N)}}.
 \end{equation}

 A multipartite system is bi-separable under a specific bi-partition if and only if the bi-partition can be written as:
 \begin{equation}
  \rho=\sum_{i}{p_i{\ \rho}_i^{(A)}\otimes\rho_i^{(B)}}.
 \end{equation}

A state is genuinely entangled if and only if it is not separable under any partition. Furthermore, separability in all partitions does not guarantee full separability of the system  \cite{dur2000classification}.

\subsection{Bipartite CMN}

\begin{Definition}
For a bipartite system, the CMN is the following scalar function  \cite{peled2021correlation}: 
\begin{equation}\label{M_h_p}
\mathcal{M}_{h,p} = \left( \sum_{R \in \binom{ \left[ d^2 \right] }{h} } \prod_{k \in R} \left[ \sigma_k \left( \mathcal{C} \right) \right]^p \right)^{1/p},
\end{equation}
where $ d = \min \left\{ d_A, d_B \right\} $, $\binom{ \left[ d^2 \right] }{h}$ is all combinations (subsets) of size $h$ from the set of $[d^2] = \{1,\ldots,d^2\}$ (as in the Cauchy–Binet formula), and $ \sigma_k \left( \mathcal{C} \right) $ denotes the $k$-th singular value of $\mathcal{C}$. For $h=d^2$ the CMN is the product of the singular values, regardless of $p$, and will be denoted by $\mathcal{M}_{h=d^2}$. Furthermore, if $d_A=d_B$ we get $\mathcal{M}_{h=d^2}=\det(\mathcal{C})$. 
\end{Definition}
The CMN is comprised of the Schatten $p$-norm of the $h$-th compound matrix of the correlation matrix. Calculation of the CMN is amount to a full tomography of the state. 

It can be shown that for a bipartite system, there exists a positive number $\mathcal{B}(d_A,d_B,h,p)$, for all $h$ and $p$, for which:
\begin{equation}\label{CMN_generic_bound}
\mathcal{M}_{h,p} \leq \mathcal{B}(d_A,d_B,h,p),
\end{equation}
if the state is separable, which gives us an entanglement detection criterion (see Fig~\ref{separable_bi_fig}). $\mathcal{B}$ can be calculated for $p=1$ and $p = \infty$ (Theorems 1 and 3 in \cite{peled2021correlation}), and is shown to be a tight bound for those values of $p$ (Theorems 2 and 4 in \cite{peled2021correlation}).

\begin{figure}[h]
\includegraphics[scale=0.3]{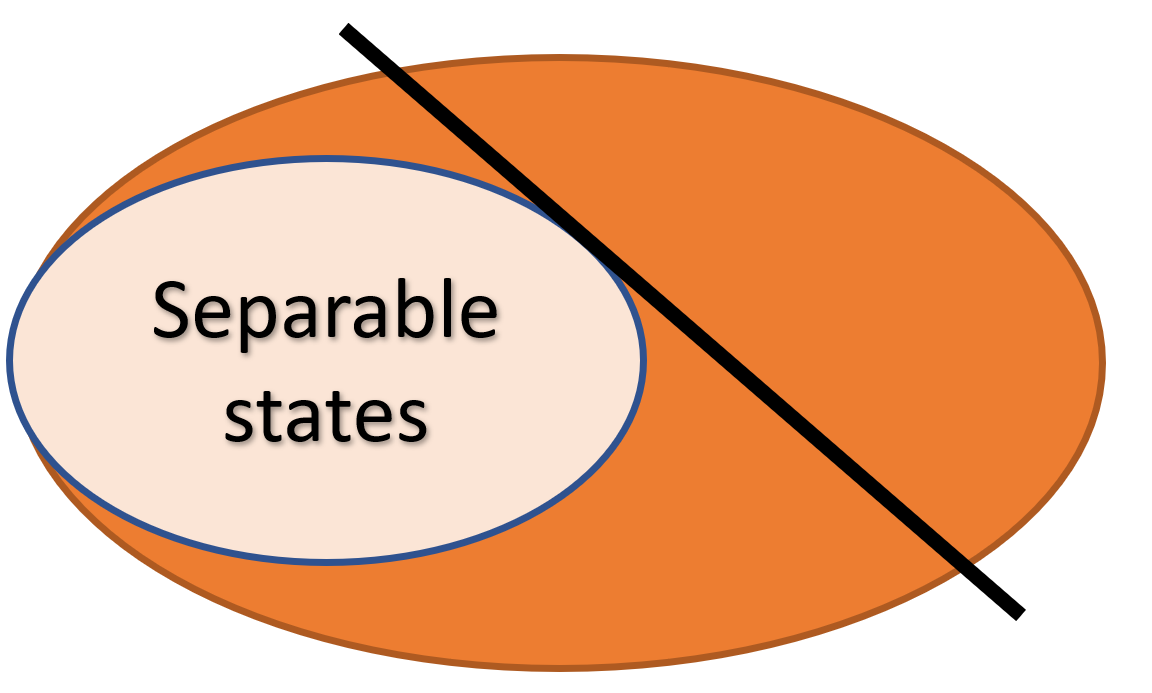}
\caption{{\bf Visualization of the CMN entanglement detection scheme.} The black ``line'' represents $\mathcal{B}$, so that states located to its right will be detected as entangled, Eq. \eqref{CMN_generic_bound}.}
\centering
\label{separable_bi_fig}
\end{figure}

\begin{comment}
\textcolor{red}{SHOULD I LEAVE THIS IN????} \textcolor{blue}{[Amit says: it's up to you, but in the current state of the manuscript I do feel like this part interrupts with the flow of the presented ``story''.]} \textbf{Side-note }- in \cite{peled2021correlation}, the following theorem was suggested, but a proof was not supplied:
\begin{Theorem}
$\forall t \in [d], \mathcal{M}_{h=t^2,p} \neq 0$ iff the state’s pure-state-Schmidt rank is at least t. Where The Schmidt rank is the number of strictly positive Schmidt coefficients counted with multiplicity (Schmidt coefficients are non negative).
\end{Theorem}

This theorem with $t = d =  \min{\{d_A,d_B\}}$, narrows down to: the product of all the singular values (the CMN in this case) is not zero if and only if the state’s pure-state Schmidt rank is at least $d$. Note that this is the same as the second part of following lemma:

THE SINGULAR VALUES ARE NOT EXACTLY THE SCHMIDT COEFFICIENTS. NEED TO FIX
FROM THE CMN FOMULA IT IS PRETTY CLEAR THAT THE THEOREM IS TRUE FOR OPERATOR SCHMIDT IFF RANK IS $T^2$, THE NONO-TRIVIAL PART IS TO TURN THE OPERATOR SCHMIDT TO THE PURE STATE SCHMIDT
(17) SHOWS THE SCHMIDT COEFFICIENTS 

\end{comment}

\section{CMN bounds for the detection of multipartite entanglement}
\subsection{The CMN function for multipartite states}
 For example, consider a three-party correlation tensor $\mathcal{C}^{ijk}$, which represents the correlations $\braket{A_i \otimes B_j \otimes C_k}$. The tensor slices are (see Figure~\ref{tensor_slices}):
\begin{equation}\label{tensor_slices_matrix}
\mathcal{C}^{ijk}=
 \begin{bmatrix}
A_1B_1C_1 & A_1B_2C_1\\
A_2B_1C_1 & A_2B_2C_1 

\end{bmatrix} ,
\begin{bmatrix}
A_1B_1C_2 & A_1B_2C_2 \\
A_2B_1C_2 & A_2B_2C_2 

\end{bmatrix}.
\end{equation}
The tensor can be flattened over its $C$ mode, yielding the matrix:
\begin{equation}
\mathcal{C}^{i/jk}=
 \begin{bmatrix}
A_1B_1C_1 & A_1B_2C_1 & A_1B_1C_2 & A_1B_2C_2 \\
A_2B_1C_1 & A_2B_2C_1 & A_2B_1C_2 & A_2B_2C_2.
\end{bmatrix}  .
\end{equation}
\begin{Lemma}
$\mathcal{C}^{i/jk}$ is the matrix representing the correlations between $A$ and $(B\otimes C)$. Indicating that tensor flattening is a means to realise a state partition. This holds true for higher order correlation tensors.
\end{Lemma}

Meaning that in order to calculate the CMN for a multipartite state, we consider all possible matricizations of the correlation tensor:
\begin{equation}\label{M_h_p}
\mathcal{M}_{h,p} = \left( \sum_{R \in \binom{ \left[ d^2 \right] }{h} } \prod_{k \in R} \left[ \sigma_k \left( \mathcal{C}_{flat} \right) \right]^p \right)^{1/p},
\end{equation}
wherein $\mathcal{C}_{flat}$ denotes any matricization of the correlation tensor.

\begin{figure}[h]
\includegraphics[scale=0.5]{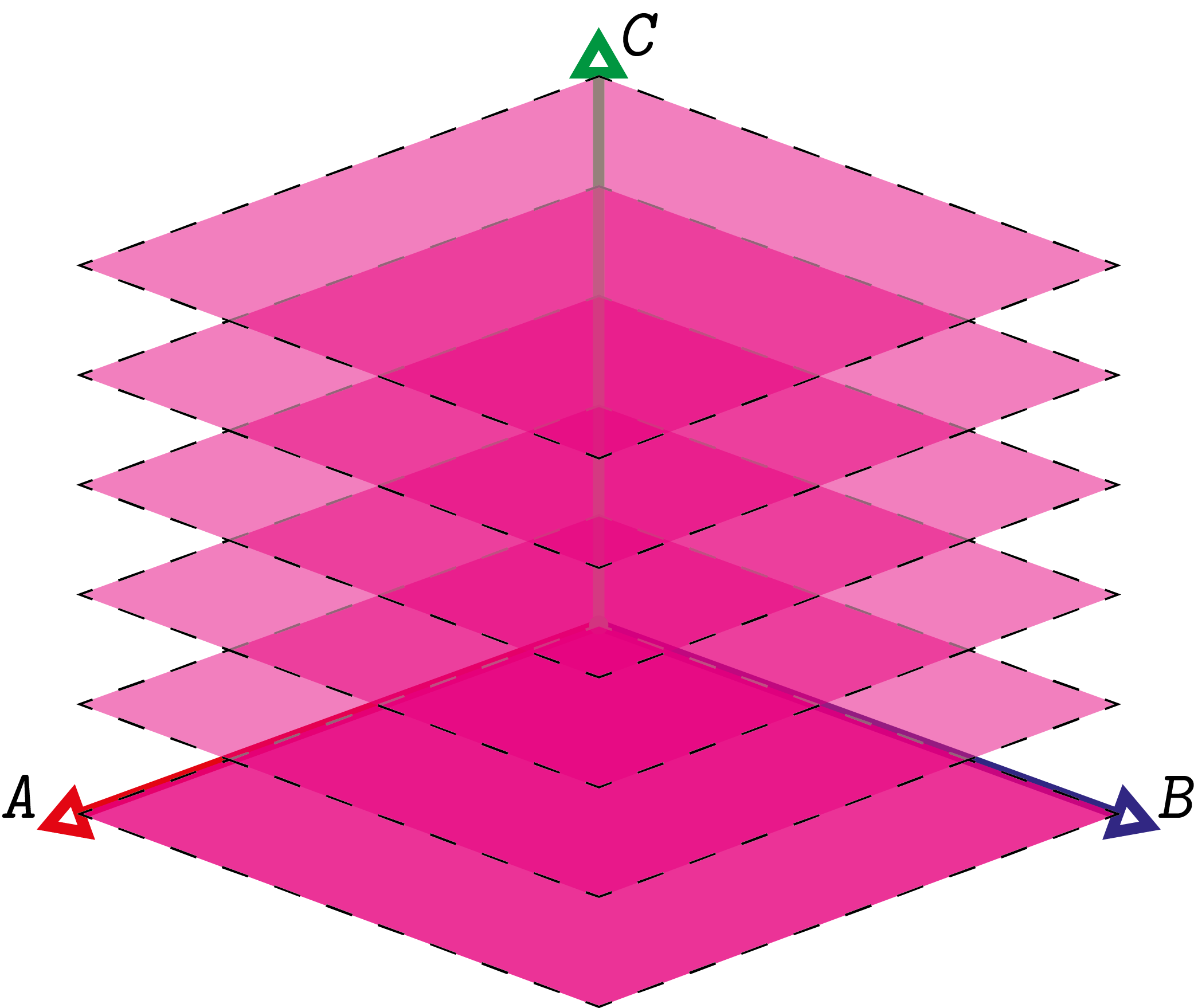}
\caption{{\bf Illustration of the tripartite correlation tensor.} The tensor slices in Eq. (\eqref{tensor_slices_matrix}) are represented by color.}
\centering
\label{tensor_slices}
\end{figure}

\subsection{Bounds on bi-separable states}\label{CMN_bi_bounds_section}

Let us assume the parties choose their orthonormal observables such that: $ A_1 = \frac{1}{\sqrt{d_A}} \mathbbm{1}_{d_A}$ (same for the other parties), i.e. the trivial measurements. This implies that all the other observables: $\{A_i\}_{i=1}^{d_A^2}$,  are traceless (because the observables are orthonormal, each inner product with $\mathbbm{1}$ should yield 0). Given this assumption, we are motivated to define the following:

\begin{Definition}
 A state in Filter normal form (FNF) is a state in which any single-party traceless observable has vanishing expectation values.  In the bipartite case, which may be a multipartite state under a bi-partition, this becomes: $  \braket{ A_i \otimes \mathbbm{1} }  =  \braket{\mathbbm{1} \otimes B_j } = 0$.
\end{Definition}

\begin{Lemma}\label{lemma_FNF}
A mixed multipartite state can be brought into a normal
form by stochastic local operations and classical communication (SLOCC), where the normal
form has all local-density operators proportional to the identity and is unique up to local unitaries \cite{verstraete2003normal} (Theorem 4).
\end{Lemma}

For a multipartitie system, bi-partitioned into parties $A$ and $B$, the bipartite CMN bounds \cite{peled2021correlation} can be used for bi-separability detection (see Figure~\ref{separable_tri_fig}): 

\begin{Theorem}\label{bi-sep-1}
	For a a bi-separable state under the bi-partition $A/B$ in FNF, and $h \geq \sqrt{d_Ad_B}$. The following holds true:
	\begin{equation}\label{cmn_tri_bound_inf}
	\mathcal{M}_{h,p=\infty} = \prod_{k=0}^{h-1}\sigma_k(\mathcal{C}) \leq \frac{1}{\sqrt{d_Ad_B}} \left[ \frac{d_A-1}{d_A \left(h-1\right)} \frac{d_B-1}{d_B \left(h-1\right)} \right]^{\frac{h-1}{2}} .
	\end{equation}
	Note that $h \leq d^2$ , $d \defeq \min \left\{ d_A, d_B \right\}$ 
\end{Theorem}

\begin{Theorem}\label{bi-sep-2}
	Assume $D \defeq \max \left\{ d_A, d_B \right\},d \defeq \min \left\{ d_A, d_B \right\} ,D\leq d^3 $, and $h >1$. Then, for any bi-separable state under the bi-partition $A/B$ in FNF:
	\begin{equation}\label{cmn_tri_bound_1}
	\mathcal{ M}_{h,p=1} = S_h \left( \sigma_1, \ldots, \sigma_{d^2} \right) \leq S_h \left( \alpha, \frac{\beta}{d^2-1}, \ldots, \frac{\beta}{d^2-1} \right)
	\end{equation}
	where $ \alpha \defeq 1/\sqrt{d_Ad_B} $, $ \beta \defeq \sqrt{ \frac{d_A-1}{d_A} \frac{d_B-1}{d_B}} $, and $S_h$ is the $h$-th elementary symmetric polynomial in $d^2$ variables:
	
	\begin{equation}
S_h \left( x_1, \ldots, x_{d^2} \right) \defeq \sum_{R \in \binom{[d^2]}{h}} \prod_{k \in R} x_k ,
\end{equation}
\end{Theorem}

If one of the bounds is violated, the state is entangled under that bi-partition, otherwise, the bounds yield no information.

\begin{Lemma}\label{sat_lemma}
Using Theorems 2 and 4 in \cite{peled2021correlation}, the bounds presented in this section are tight, meaning that no smaller bound can exist. The condition $h \geq \sqrt{d_Ad_B}$ in Theorem \ref{bi-sep-1}, and $D\leq d^3$ in Theorem \ref{bi-sep-1} are needed for the tightness of the bounds.
\end{Lemma}

\begin{figure}[h]
\includegraphics[scale=0.3]{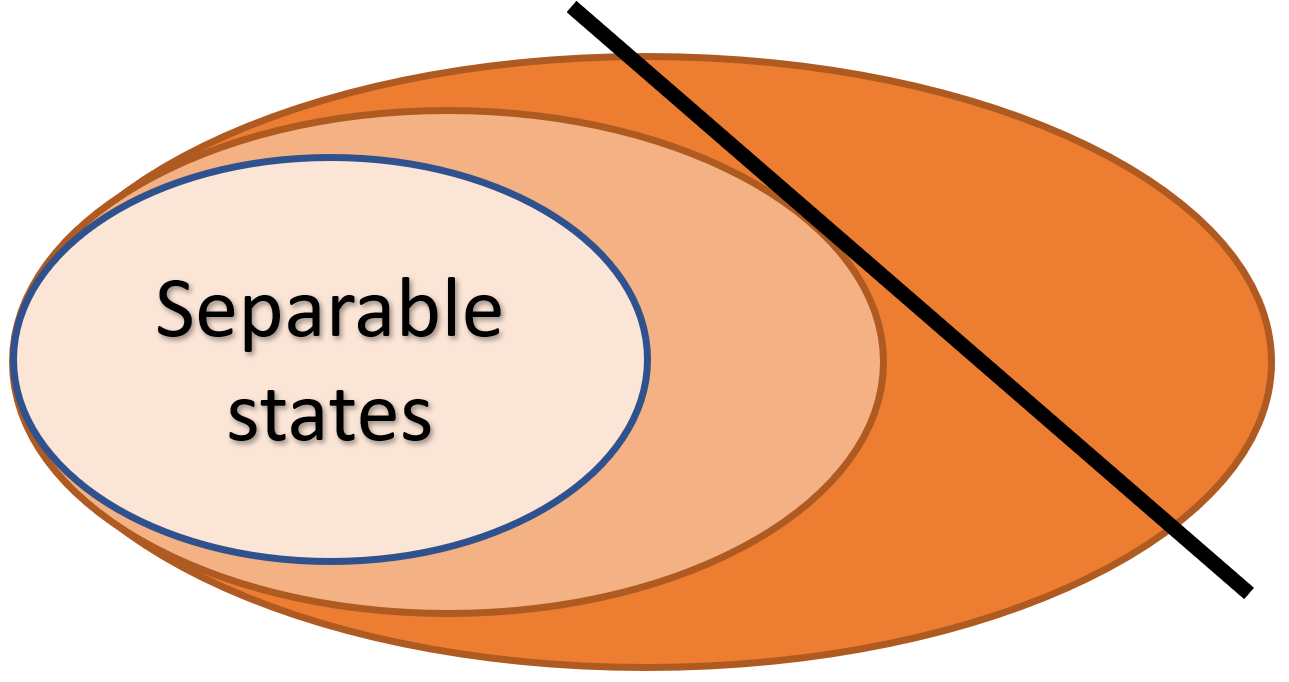}
\caption{{\bf Visualization of the CMN entanglement detection scheme for multipartite systems.} The black ``line" represents $\mathcal{B}$  where states right of the black will be detected as bi-entangled (for a specific bi-partition). The midway orange zone represents bi-separable states. Note that since
there exists a bi-separable state that saturates the bounds, the black line is tangent to the bi-separable states zone.}
\centering
\label{separable_tri_fig}
\end{figure}

\subsection{Symmetric bi-separable state which saturate the CMN bound}\label{bi_separable_sat}
In order to discuss our construction, we must first define the following:
\begin{Definition}
A symmetric, informationally complete, positive operator-valued measure (SIC-POVM) of dimension $d$ can be thought of as a group of $d^2$ pure states which are equally spaced (a regular, coherent, degree-$1$ quantum design with $r = 1$ and $d^2$ elements). For example, if $d=2$ the SIC-POVM state yields a tetrahedron on the Bloch sphere, as can be seen in Figure~\ref{SIC-POVM_figure}.
\end{Definition}

\begin{figure}[h]
\includegraphics[scale=0.5]{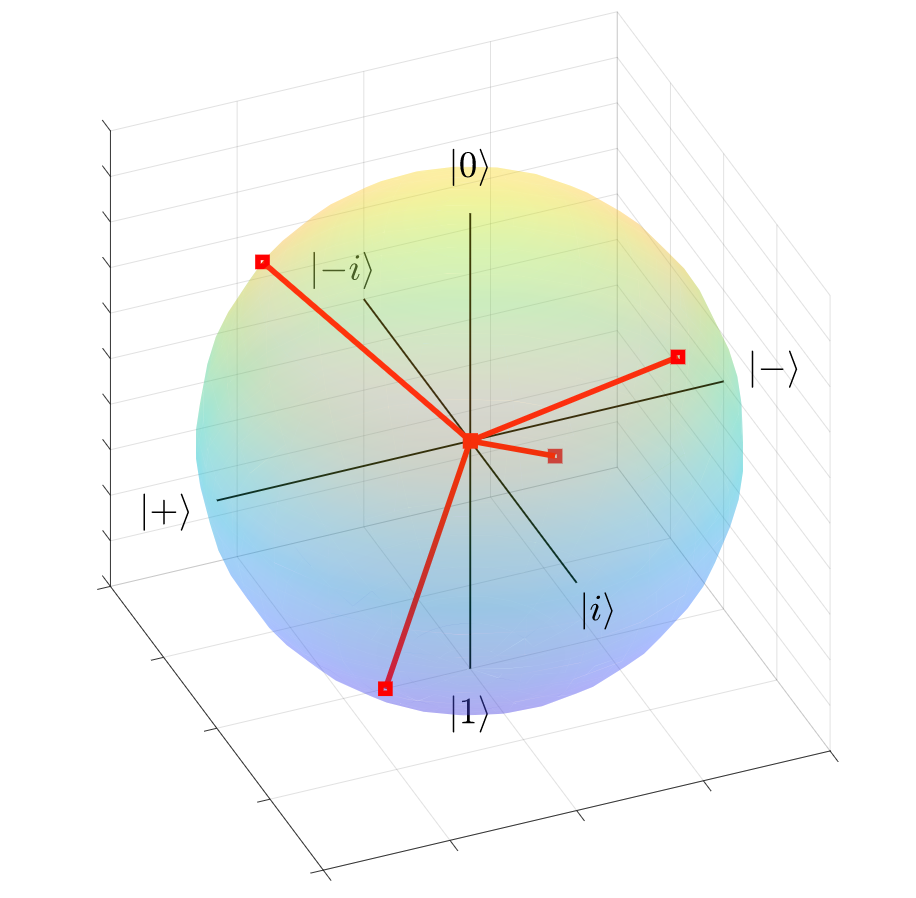}
\caption{{\bf SIC-POVM for $d=2$ creates a tetrahedron.} The 4 vectors are evenly spaced.}
\centering
\label{SIC-POVM_figure}
\end{figure}

Now, consider the following tripartite mixed state of three qubits:
\begin{equation}\label{sat_CMN_BI_state}
    \rho_1 = \sum_{i=1}^4 \rho^{SIC-POVM}_i\otimes \ket{ \psi^{Bell}_i} \bra{\psi^{Bell}_i},
\end{equation}

  \begin{equation}
    \rho^{SIC-POVM}_i = \frac{\mathbbm{1} + \Vec{r}_i \cdot \Vec{\sigma} }{2},
\end{equation}
with $\Vec{r}_i \in \mathbb{R}^3$ the vertices of a regular tetrahedron centered at the origin, rotated as seen in Figure \ref{SIC-POVM_figure}, where each vertex lies on the unit sphere. The states $\rho^{SIC-POVM}$ and $\ket{ \psi^{Bell}_i} \bra{\psi^{Bell}_i}$ are explicitly given in Appendix \ref{states_for_comparison}.

Assume that party A holds the SIC-POVM state and parties B and C share the Bell basis. Using Theorems 2 and 4 in \cite{peled2021correlation}, we know that this state saturates the CMN for the bi-partition $A/BC$ for $p=1,\infty$. One may calculate the correlation tensor (we have used the identity operator and the Pauli matrices as the observables) and see that it is diagonal, meaning that the correlation tensor is invariant under permutations. Thus, the state $\rho_1$ saturates the  CMN bound for bi-separable states \emph{for every bi-partition}. This notion is also true for the multipartite entanglement detection scheme presented by de Vicente and Huber (the multipartite dVH criterion) \cite{de2011multipartite}, as shown in Sec. \ref{compare_bell}, wherein we calculate the CMN and multipartite dVH bounds for $\rho_1$.

Permutation invariant states such as $\rho_1$ are rather prevalent in quantum mechanics, and are important for representation of many-body bosonic states, and security of quantum key distribution protocols \cite{leverrier2013security,sheridan2010finite}, for example. Thus, replicating the states presented here for higher dimensional $N$-party system systems is of interest, using higher dimensional SIC-POVM and Bell states \cite{ccorbaci2016construction,fujii2001generalized,karimipour2002quantum}.

\subsection{CMN bounds for fully-separable states}
In a similar manner to the above, let us define the following:
\begin{Definition}
 A state in strong FNF (SFNF) is a state in which any correlation involving the unity observable, excluding the main tensor vertex: $\braket{A_0\otimes B_0 \otimes \ldots \otimes N_0}$, is zero. See Figure~\ref{tensor_figure}. This is a state where the only non-zero correlations are between all the parties.
\end{Definition}
As far as our knowledge goes, a result such as Lemma \ref{lemma_FNF} does not exist for SFNF.

Now, we may bound the CMN for the fully-separable case, thus detecting any form of entanglement:

\begin{figure}[h]
\includegraphics[scale=0.25]{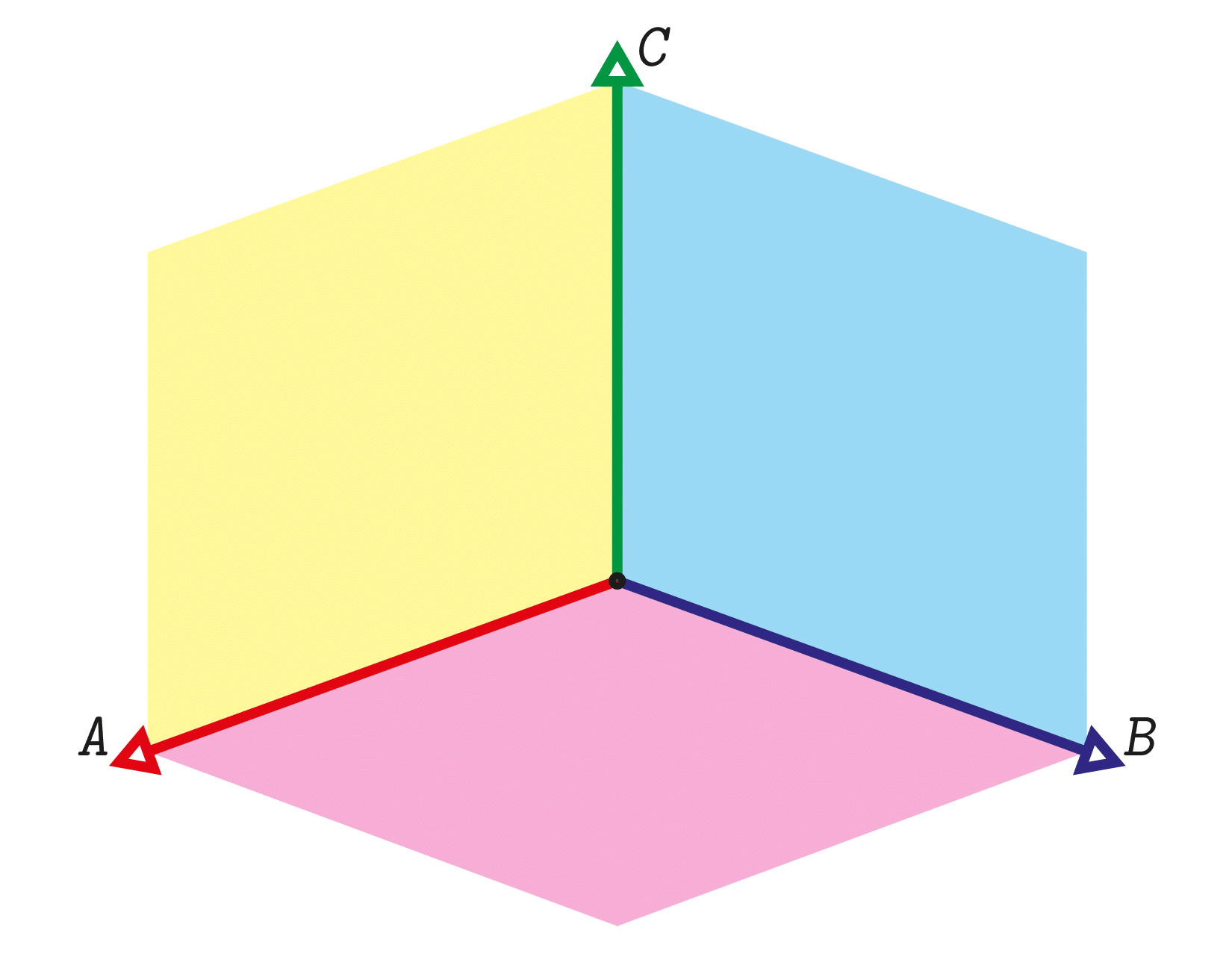}
\caption{{\bf Faces of the tripartite correlation tensor.} A tripartite state is said to be in SFNF, if the correlations on its faces are all zero - excluding the vertex at the origin, which has a fixed nonzero value. The
tensor faces are the bipartite correlation matrices of the reduces states (particle trace)}
\centering
\label{tensor_figure}
\end{figure}

 \begin{Theorem}\label{mult_full_1}
 For a matricized correlation tensor under any bi-partition $A/B$, assume $h \geq \prod_{i=1}^N\sqrt{d_i-1}$. Then, for a fully-separable state in SFNF, it holds that:
 \begin{equation}
\mathcal{M}_{h,p=\infty} = \prod_{k=0}^{h-1}\sigma_k(\mathcal{C}) \leq \alpha \left(\frac{\beta}{h-1} \right)^{h-1}.
 \end{equation}
 Where $\alpha = \prod_{i=1}^N\frac{1}{\sqrt{d_i}}\; , \;  \beta = \prod_{i=1}^N\sqrt{\frac{d_i-1}{d_i}}$. Note that $h \leq d^2$ , $d \defeq \min \left\{ d_A, d_B \right\}$ .
 \end{Theorem}
 
\begin{comment}
 \begin{Corollary}
  For $d = \min \left\{ d_A , d_B \right\}$ and $h=d^2$, the CMN is the product of all singular values and Theorem \ref{mult_full_1} has the special case of:
 
 \begin{equation}
 \mathcal{M}_{h=d^2} = \prod_{k=1}^{d^2}\sigma_k \leq \alpha \left(\frac{\beta}{d^2-1} \right)^{d^2-1}.
 \end{equation}
 \end{Corollary}
\end{comment}
 
\begin{Theorem}\label{mult_full_2}
For a matricized correlation tensor under any bi-partition $A/B$,
assume $D \defeq \max \left\{ d_A , d_B \right\},d \defeq \min \left\{ d_A , d_B \right\}$. Then, for any fully-separable state in SFNF:
 \begin{equation}
 \mathcal{M}_{h,p=1} = S_h \left( \sigma_1, \ldots, \sigma_{d^2} \right) \leq  S_h \left(\alpha, \frac{\beta}{d^2-1}, \ldots, \frac{\beta}{d^2-1} \right)
 \end{equation}
 	Where $\alpha = \prod_{i=1}^N\frac{1}{\sqrt{d_i}}\; , \;  \beta = \prod_{i=1}^N\sqrt{\frac{d_i-1}{d_i}}$, and $S_h$ is the $h$-th elementary symmetric polynomial in $d^2$ variables, shown earlier.
\end{Theorem}

Recall that the CMN is calculated using tensor matricization, and unlike in the bi-separable case, all matricizations must be checked. Even if only one of the matricizations breaks the bound, the state is not fully-separable.

Note that for a bipartite state, the theorems presented here are equivalent to the bi-separable theorems presented in Sec. \ref{CMN_bi_bounds_section}.

The proofs to the bounds appearing in this section can be found in Appendix \ref{proofs_fully_bounds}.

\section{Relation to other multipartite entanglement detection schemes}

We wish to compare our results to the multipartite dVH criterion \cite{de2011multipartite}, being the logical thing to do as the CMN bounds were derived from the bipartite \cite{vicente2007separability} and multipartite dVH criterion. 

Under the normalization of operators presented in this work, if a state is fully-separable the multipartite dVH criterion states that:
\begin{equation}\label{dV_full_sep}
    \sum_i\sigma(\mathcal{W}_{flat}) \leq \prod_{j=1}^N\sqrt{\frac{d_j-1}{d_j}},
\end{equation}
for every matricization. and for a three-qubit bi-separable state (for a specific bi-partition):
\begin{equation}\label{dV_bi_sep}
    \sum_i\sigma(\mathcal{W}_{flat}) \leq \sqrt{\frac{3}{8}}
\end{equation}
Wherein $\mathcal{W}$ is the ``interior'' correlation tensor, meaning that every correlation involving the unity 
observable is removed, as can be seen in Figure \ref{tensor_interior}. This is in fact the tensor of non-zero correlations for a state on SFNF.

\begin{figure}[h]
\includegraphics[scale=0.45]{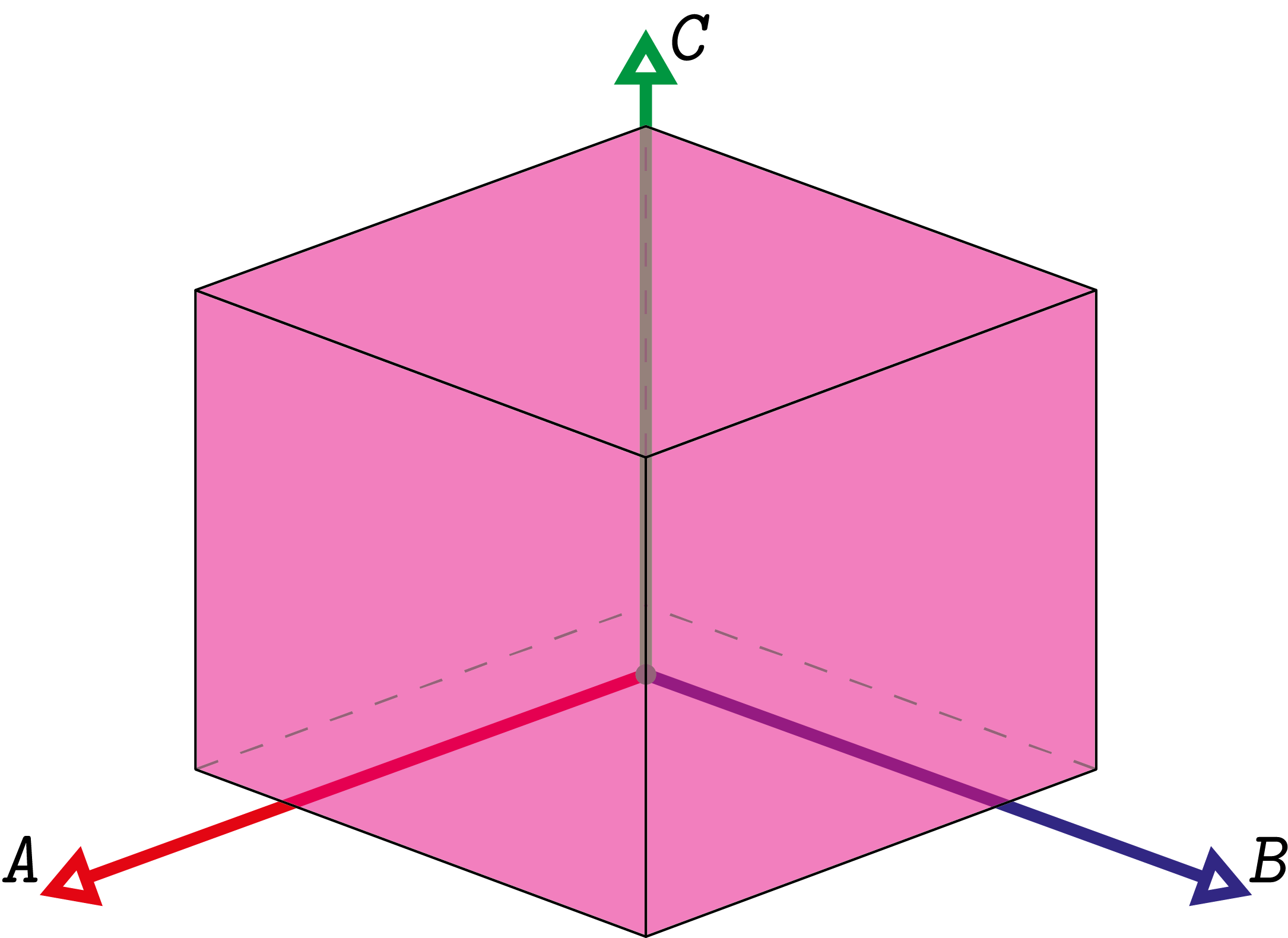}

\caption{{\bf Removing the faces of the tripartite correlation tensor.} This leaves us with the ``interior'' correlation tensor used in the dVH multipartite entanglement criterion.}
\centering
\label{tensor_interior}
\end{figure}

Under SFNF assumptions, for fully-separable states the CMN and multipartite dVH criterion are equivalent. If a state does not hold SFNF, the CMN is not applicable while the multipartite dVH is.
For bi-separable states, the CMN detection power may be better, due to the fact that it disregards less correlations, as will be shown in the following sub-sections.

\subsection{SIC-POVM with Bell}\label{compare_bell}
Consider three parties sharing the state presented in Sec. \ref{bi_separable_sat}. For this state we have shown that each matricization is the same because the correlation tensor is symmetric.
\begin{itemize}
\item 
\textbf{Multipartite dVH criterion}: Using Eq. \ref{dV_full_sep}, the criterion detects that the state is not fully separable:$\sqrt{3/8} \nleq 2^{-1.5}$. Furthermore, the bound for bi-separable states is saturated, as we receive an equality in Eq. \ref{dV_bi_sep}.
\item
\textbf{CMN criterion}: Using Theorem \ref{mult_full_1} and assuming that $h=d^2$, the CMN is the product of all singular values and we detect that the state is not fully separable: $1/(64\cdot 3 \sqrt{3}) \nleq 1/1728.
$

Using Theorem \ref{bi-sep-1} and assuming that $h=d^2$, the CMN is the product of all singular values and saturation of the bound is received, as was shown is Sec. \ref{bi_separable_sat}.

Using the CMN with $p=1$ would give the same entanglement detection power; Theorem \ref{mult_full_2} detects that the state is not fully-separable: $1/(8\sqrt{3}) \nleq 1/24$, and we get saturation of the bound in Theorem \ref{bi-sep-1}.

\end{itemize}

For the ``SIC-POVM with Bell" state ($\rho_1$), the multipartite dVH and the CMN are equivalent in entanglement detection power.
\subsection{breaking the dVH bound using the CMN}

In the multipartite dVH criterion, all unity correlation are disregarded, as was previously discussed. Those values can hold important correlations that may turn the tide on detecting entanglement. Thus, consider the tripartite qubit state $\rho_2$, given explicitly in Appendix \ref{states_for_comparison}.

\begin{itemize}
\item 
\textbf{multipartite dVH criterion}: Using Eq. \ref{dV_full_sep}, the criterion detects that the state is not fully separable: $0.4982 \nleq 2^{-1.5}$. But due to the fact that the multipartite dVH criterion disregards any correlation involving the unity observable, Eq. \ref{dV_bi_sep}, which is the bound for bi-separable states, yields no information: $0.4982 \leq \sqrt{3/8} \; , \; 0.4982 \leq \sqrt{3/8} \; , \; 0.4784 \leq \sqrt{3/8}$. Where we calculated the bound for each matricization.
\item
\textbf{CMN criterion}: $\rho_1$ does not obey the SFNF assumptions, and thus we are unable to use theorems \ref{mult_full_1},\ref{mult_full_2} for fully-separable states. For two out of three bi-partitions, $\rho_1$ holds the FNF assumptions needed for the bounds on bi-separable states (for the third bi-partition the state can be brought into FNF using Lemma \ref{lemma_FNF}.). Thus, Using Theorem \ref{bi-sep-1} and assuming that $h=d^2$, the CMN is the product of all singular values, and is able to detect bi-entanglement in two out of the three bi-partitions: $ 3.0549 \cdot 10^{-3} \nleq  1/(64\cdot 3 \sqrt{3}) \; , \;   3.0549 \cdot 10^{-3} \nleq  1/(64\cdot 3 \sqrt{3})$. In a similar manner, we could have used $p=1$.
\end{itemize}
Note that although the CMM is able to detect bi-entanglement in this case, some correlations are still disregarded, which gives motivation for the next section. 

\begin{comment}
\subsection{States undetectable by the CMN and the multipartite dVH}

Here, we continue our discussion from the last subsection on the effects of disregarding unity correlations. Thus, consider the tripartite qubit state $\rho_2$, given explicitly in Appendix \ref{states_for_comparison}. For our construction of the correlation tensor, the following bounds arise due to the bipartite dV criterion for each matricization:
\begin{equation}
    0.6169
\leq \sqrt{\frac{3}{8}} = 0.61237 \;,\;  0.6169 \leq0.61237 \;,\; 0.4642 \leq0.61237
\end{equation}
The bipartite dV criterion detects bi-entanglement in 2 out of three cases. Note that we are unable to use the bipartite CMN due to the fact that the state does not satisfy the required FNF assumptions.

The multipartite dVH bounds hold:
\begin{equation}
    0.9533 \leq 1 \;,\; 0.9533 \leq 1 \;,\;
\end{equation}
thus the multipartite dVH criterion is unable to detect any non-separability.
\end{comment}

\section{Improving the detection scheme}

In the last section, we have discussed the importance of the reduced state correlations, which are usually disregarded in one way or another in entanglement detection schemes using correlation tensors. Combined with the fact that a similar construction of Lemma \ref{sat_lemma} does saturate of the bounds on fully-separable states, presumably due to the same disregarded correlations, hints to some form of information loss which may turn the tide on entanglement detection. Thus, we offer to use the discarded correlation in the following improvement of the detection scheme: In order to detect that a state is not fully-separable, one may separately inspect the ``interior" correlation tensor (using the multipartite dVH criterion, for example) and the reduced state correlations. In the tripartite case, the reduced state correlations are just the three tensor faces (figure \ref{tensor_figure}), which can be with inspected using any bi-partite entanglement detection scheme (using the bi-partite CMN, for example). In the general N-party case, the reduced state correlations are tensors, and the process is to be applied recursively. Note that same line of thought can be used in order to improve detection of bi-entanglement.

The entanglement detection scheme presented in Reference \cite{sarbicki2020family} for example, does take into account the reduced state correlations, and a comparison is in order between the schemes.

\section{Quantum Discord}

\subsection{Validity of original results}
For the bipartite CMN \cite{peled2021correlation}, motivated by the definition and expression for geometric quantum discord derived in~\cite{luo2010geometric}, the following measure for discord w.r.t. Alice's subsystem was suggested:
\begin{equation}\label{Discord_CMN}
    \mathcal{D}^A_{h,p} \left( \rho \right) = \left[ \mathcal{M}_{h,p} \left( \rho \right) \right]^p  - \max_{\Pi^A \in M \left( A \right) } \left[ \mathcal{M}_{h,p} \left( \Pi^A \left[ \rho \right] \right) \right]^p ,
\end{equation}
where the maximization goes over all projective measurements on Alice's subsystem $ \Pi^A = \left\{ \Pi_i \right\}_{i=1}^{d_A} $, and $ \Pi^A \left[ \rho \right] $ is the state obtained from $\rho$ by performing the measurement $ \Pi^A $ and obtaining the appropriate ensemble of the projections $\Pi_i$ (i.e., the state is measured but not ``collapsed''). The following result suggests that $ \mathcal{D}^A_{h,p} $ may be thought of as a measure for discord:
\begin{Theorem}\label{discord_theorem}
    For any state $\rho$ and for any value of $h, p$, $ \mathcal{D}^A_{h,p} \left( \rho \right) \geq 0 $; and $ \mathcal{D}^A_{h\leq 2,p} \left( \rho \right) = 0 $ iff $ \mathcal{D}_G^A \left( \rho \right) = 0 $.
\end{Theorem}

The theorem above is also valid for the multipartite case, when considering a system under a bi-partition. The original proof of this theorem can be easily adapted for the multipartite case, but one can be convinced by considering that Alice's discord measure compared to Bob's is no different when dividing Bob into several parties.

\subsection{Multipartite discord}

As the CMN captures the notion of discord in bipartite systems, it is intriguing to consider the same notion for multipartite system, In which global quantum discord is defined as: \cite{rulli2011global,xu2012geometric}:
 \begin{equation}
  \mathcal{D}^{G}(\rho_{A_1 \ldots A_N}) \defeq \sum_{i=1}^N S(\rho_{A_i}) - S(\rho_{A_1 \ldots A_N}) - \max_{\Pi}\left[\sum_{i=1}^N S(\Pi^{A_i}\rho_i)-S(\Pi\rho)\right] .
 \end{equation}
  Wherein $S$ is the von Neumann entropy, $S_{A_i}$ is the von Neumann entanglement entropy for reduced states, and $\Pi = \Pi^{A_1}\otimes \ldots \otimes \Pi^{A_N}$ is a measurement across the entire system. The maximization is taken in order to remove the dependence on the measurement and to capture only the non-local (quantum) correlations. This expression was shown to be non-negative, and to be zero iff the state is a classical one \cite{xu2012geometric}, as expected from quantum discord. Thus, we are motivated to present the following measure for global quantum discord, using the CMN:
\begin{equation}\label{Discord_CMN_multipartite}
    \mathcal{D}^{Global}_{h,p} \left( \rho \right) = \left[ \mathcal{M}_{h,p} \left( \rho \right) \right]^p  - \max_{\Pi } \left[ \mathcal{M}_{h,p} \left( \Pi \left[ \rho \right] \right) \right]^p.
\end{equation}

  The following result implies that the CMN's is a measure for global quantum discord:
\begin{Theorem}\label{discord_theorem_multipartite}
    For any state $\rho$ and for any value of $h, p$, $ \mathcal{D}^{Global}_{h,p} \left( \rho \right) \geq 0 $; and $ \mathcal{D}^{Global}_{h\leq 2,p} \left( \rho \right) = 0 $ iff $ \mathcal{D}_G \left( \rho \right) = 0 $.
\end{Theorem}
Due to the fact the $\Pi$ is comprised of local projective measurement, it can be decomposed into several measurements, one for each party. Thus, the logic behind the bipartite case can be used to claim that if after each measurement the CMN reduces in value, then measuring $\Pi$ would lead to $\mathcal{M}_{h,p} \left( \rho \right) \geq \mathcal{M}_{h,p} \left( \rho' \right) $, for all $h,p$, in which $\rho'$ denotes the state after being measured by $\Pi$. The complete proof for this theorem appears in Appendix \ref{Proof_global_discord}. In \cite{xu2012geometric}, it was presented that quantum discord can be computed using the correlation tensor, we have shown that the notion of discord is preserved when matricizing the correlation tensor.

  \begin{comment}
    A geometric definition can be had for global quantum discord, as was in the bipartite case:
\begin{equation}
     \mathcal{D}^{GG}(\rho_{A_1 \ldots A_N}) \defeq \min_{\sigma_1,\ldots, \sigma_N}\left[\tr{\rho_{A_1 \ldots A_N}-\sigma_{A_1 \ldots A_N}}^2|\mathcal{D}^G\left(\sigma_{A_1 \ldots A_N}\right)=0\right].
\end{equation}
  \end{comment}

\section{Conclusions}

Our goal in this work was to detect multipartite entanglement via matricizations of correlation tensors. We generalized a bipartite entanglement detector, the CMN, reevaluating and extending it to multipartite systems. Our detection scheme is in the form of bounds, which through the same endeavour, we have managed to saturate in the case of bi-separable states. We further observed information loss in known
methods for detecting entanglement from correlation tensors, which fail to take into account reduced state (partial trace) correlations. Thus, we have presented a method which seeks to overcome this issue by using multiple bounds on correlations of the reduced states. However, we do not know whether it is possible to find a single bound that would have the same detection capabilities as our multi-stage scheme. Such a parameter, if exists, would necessarily have to consider all correlations simultaneously; the method of using local filtering does not apply, since it is not always possible to transform a state into SFNF using SLOCC. Overcoming these issues is a possible direction for future research.

Moreover, our discussion led us to present a permutation invariant state which saturates the bound on bi-separable states. It is fundamentally interesting to ask whether one may generalize this notion into a family of such states for $N$-party systems and for higher dimension. Such a construction could also prove useful in fields and applications where permutation-invariant states arise naturally, for example many-body bosonic systems, and security of quantum key distribution protocols.

Furthermore, using the same logic of tensor matricization, we have shown that the CMN is a measure for global quantum discord. It was already known that multipartite full-separability can be characterized via state partitions, specifically when using matricized correlation tensors. Here, we showed a similar idea for discord, as in that the notion of discord is preserved when matricizing correlation tensors, compared to other measures which use the entire tensor. It is interesting to see what properties of a multipartite state are preserved under state partitions, which, as we have shown, require only dealing with matrices instead of higher-rank tensors.

Finally, having generalized the CMN for multipartite systems, it might be interesting to ask if it can be generalized for continuous variable systems. In such generalization, several questions arise: would it be easier to take into account reduced state correlations? Can our measure of global quantum discord be adapted to this setting?

\section[\appendixname~\thesection]{Appendix - Proofs of theorems}
\subsection[\appendixname~\thesubsection]{Proving the bounds on fully-separable states}\label{proofs_fully_bounds}

Under the SFNF assumption, the singular values of $\mathcal{C}_{flat}$ and $\mathcal{W}_{flat}$ (of the multipartite dVH criterion \cite{de2011multipartite}) will be the same, up to an extra singular value:
\begin{equation}\label{dV_full_Eq}
   \sigma(\mathcal{C}_{flat}) = \left[\sigma_0 = \prod_i \left(\frac{1}{\sqrt{d_i}} \right), \sigma_1, \ldots, \sigma_{d^2-1} \right] \; , \;\sigma(\mathcal{W}_{flat}) =  \left[\sigma_1, \ldots, \sigma_{d^2-1} \right].
\end{equation}
This is due to the fact that the singular values do not change under substitution of rows/columns and removal of rows/columns of zeros. Furthermore, the extra singular value is due to the main tensor vertex: $\braket{A_0\otimes B_0 \otimes \ldots \otimes N_0}$, which is not zero under SFNF. 

Thus, we may bound the CMN for different cases:

\begin{proof}[Proof of Theorem \ref{mult_full_1}]
We wish to claim that $\sigma_0 = \prod_i \left(\frac{1}{\sqrt{d_i}} \right)$ is among the $h$ largest singular values, we may use Eq. \eqref{dV_full_sep} to claim that:
\begin{equation}
\sigma_{h} \leq \frac{1}{h}\prod_{j=1}^N\sqrt{\frac{d_j-1}{d_j}},
\end{equation}
because the singular values are in descending order, so in the ``worst case" they are equal. Thus, for our bounds, we claim that:
\begin{equation}
\sigma_0 = \prod_{i=1}^N \left(\frac{1}{\sqrt{d_i}} \right) \geq \frac{1}{h}\prod_{i=1}^N\sqrt{\frac{d_i-1}{d_i}} \Rightarrow h \geq \prod_{i=1}^N\sqrt{d_i-1}.
\end{equation}

Thus, a multipartite state in SFNF, under the any bi-partition $A/B$, holds that:
\begin{equation}
\begin{split}
& \mathcal{M}_{h,p=\infty} = \prod_{k=0}^{h-1} \sigma_k \left( \mathcal{C}_{flat} \right)  = \prod_i \left(\frac{1}{\sqrt{d_i}} \right) \prod_{k=1}^{h-1} \sigma_k \left( \mathcal{W}_{flat} \right) \leq
\\
&\leq \prod_i \left(\frac{1}{\sqrt{d_i}} \right) \left( h-1 \right)^{-\left( h-1 \right)} \left[ \sum_{k=1}^{h-1} \sigma_k \left( \mathcal{W}_{flat} \right) \right]^{h-1} \leq 
\\
&\leq \prod_i \left(\frac{1}{\sqrt{d_i}} \right) \left( h-1 \right)^{-\left( h-1 \right)}\left[\prod_i\left(\frac{d_i-1}{d_i}\right)\right]^{\frac{h-1}{2}}.
\end{split}
\end{equation}
 wherein $C_h$ is the compound matrix, and we have used the inequality of arithmetic and geometric means in the fourth transition and Eq. \eqref{dV_full_sep} the the last transition. Using our definitions for $\alpha$ and $\beta$, we are done.
\end{proof}

\begin{proof}[Proof of Theorem \ref{mult_full_2}]For a multipartite state in SFNF, under the any bi-partition $A/B$, We obtain:

\begin{align}\label{bounding_p_1}
 & \mathcal{M}_{h,p=1} = S_h \left( \alpha, \sigma_1, \ldots, \sigma_{d^2-1} \right) = \nonumber\\
 & = \alpha S_{h-1} \left( \sigma_1, \ldots, \sigma_{d^2-1} \right) + S_h \left( \sigma_1, \ldots, \sigma_{d^2-1} \right)
\end{align}
Let us denote $ s \defeq \sum_{k=1}^{d^2-1} \sigma_k $. Clearly $ s \leq \beta$. Moreover, the vectors $ \vec{\sigma} \defeq \left( \sigma_1, \ldots, \sigma_{d^2-1} \right) $ and $ \vec{e} \defeq \frac{s}{d^2-1} \left( 1, \ldots, 1 \right) $ both sum up to $s$; thus, $ \vec{\sigma} \succeq \vec{e} $ ($\succeq$ denotes majorization). Since the symmetric polynomials $ S_h $ are Schur concave, we obtain:
\begin{equation}
S_h \left( \sigma_1, \ldots, \sigma_{d^2-1} \right) \leq S_h \left( \frac{s}{d^2-1}, \ldots, \frac{s}{d^2-1} \right) .
\end{equation}
Next, we use the fact that $ S_h $ is monotonically increasing in each of its variables, alongside the inequality $s \leq \beta $, to obtain:
\begin{equation}
S_h \left( \sigma_1, \ldots, \sigma_{d^2-1} \right) \leq S_h \left( \frac{\beta}{d^2-1}, \ldots, \frac{\beta}{d^2-1} \right) .
\end{equation}
Substitution in Eq. \eqref{bounding_p_1} yields:
\begin{align}
\mathcal{M}_{h,p=1} & \leq \alpha S_{h-1} \left( \frac{\beta}{d^2-1}, \ldots, \frac{\beta}{d^2-1} \right) + S_h \left( \frac{\beta}{d^2-1}, \ldots, \frac{\beta}{d^2-1} \right) = \nonumber\\
& = S_h \left( \alpha, \frac{\beta}{d^2-1}, \ldots, \frac{\beta}{d^2-1} \right) 
\end{align}
where $\beta$ is always repeated $d^2-1$ times.
\end{proof}

\subsection[\appendixname~\thesubsection]{Proving of the CMN as a global quantum discord measure}\label{Proof_global_discord}

\begin{proof}[Proof of Theorem \ref{discord_theorem_multipartite}]
Consider a mutipartite state partitioned into a bipartite state $A$ and $B$, with the respective matricization of the correlation tensor $\mathcal{C}$. We may further consider a local measurement on one of the parties the construct $A$, for example, if $A$ consists of two parties we may measure $\mathbbm{1} \otimes \sigma_X$. The evolution of the correlation matrix under such a measurement is given by: $\mathcal{C}' =\mathcal{A}_i\mathcal{C}$, where $\mathcal{A}_i$ is a $d^2_{A} \times d^2_{A}$ matrix (the construction of $\mathcal{A}_i$ is exactly the same as in Theorem 1 of \cite{luo2010geometric}). Now, by Theorem 6.7(7) in \cite{hiai2014introduction}, for all $k \in \left\{ 1, \ldots, d_A^2 \right\} $ we have
\begin{equation}
    \sigma_k \left( \mathcal{A}_i\mathcal{C} \right) \leq \norm{ \mathcal{A}_i }_1 \sigma_k \left( \mathcal{C} \right).
\end{equation}
If we were to measure on B, the matrix multiplication would be on the right, and the the construction still holds:
\begin{equation}
    \sigma_k \left( \mathcal{C}\mathcal{B}_i \right) \leq \norm{ \mathcal{B}_i }_1 \sigma_k \left( \mathcal{C} \right).
\end{equation}

 Due to the fact that $\Pi$ consists of local projective measurement, which can be interchanged, we may continue measuring the state, where all measurements will comprise $\Pi$. Each measurement will further decrease the value of the singular values of the (post measurement) correlation matrix, yielding:
\begin{equation}
    \sigma_k \left( \mathcal{F}_A\mathcal{C} \mathcal{F}_B \right) \leq \norm{ \mathcal{F}_A }_1\norm{ \mathcal{F}_B }_1 \sigma_k   \left( \mathcal{C} \right).
\end{equation}
Wherein $\mathcal{F}_A = \prod_i\mathcal{A}_i$ , $\mathcal{F}_B = \prod_i\mathcal{B}_i$ represent the transformation undergone by the correlation matrix under $\Pi$, and $\mathcal{A}_i$, $\mathcal{B}_i$ are Commutative, as they represent non-local measurements.

 Because $ \mathcal{A}_i $ is a projection, it holds that: $ \norm{ \mathcal{A}_i }_1 = \max_j \sigma_j \left( \mathcal{A} \right) = 1 $, and the same for a measurement on $B$. Therefore, $ \sigma_k \left( \mathcal{F}_A\mathcal{C} \mathcal{F}_B \right) \leq \sigma_k \left( \mathcal{C} \right) $ for all $k$, and we conclude that $ \mathcal{M}_{h,p} \left( \rho \right) \geq \mathcal{M}_{h,p} \left( \rho' \right) $ for all $h,p$, using the fact that the CMNs are all monotonically non-decreasing w.r.t. the singular values $ \sigma_k $.
 
  Suppose $\rho$ has zero discord, which happens if and only if there exists a measurement $\Pi$ that does not disturb the state - i.e., there exists a matrices $ \mathcal{F}_A$ and  $\mathcal{F}_B $ such that $ \mathcal{F}_A\mathcal{C} \mathcal{F}_B = \mathcal{C} $. Then, for this choice of measurement, we have $ \mathcal{M}_{h,p} \left( \rho \right) - \mathcal{M}_{h,p} \left( \rho' \right) = 0 $. By the non-decreasing property for the CMN we have proven above, this is indeed the maximum, hence $ \mathcal{D}_{h \leq 2, p} \left( \rho \right) = 0 $.
  
\end{proof}

\section[\appendixname~\thesection]{Appendix - States used in this work}\label{states_for_comparison}

First, let us present the SIC-POVM with Bell state: $ \rho_1 = \sum_{i=1}^4 \rho^{SIC-POVM}_i\otimes \ket{ \psi^{Bell}_i} \bra{\psi^{Bell}_i}$. Wherein:
\begin{equation}
\begin{split}
\rho^{SIC-POVM}_1 = 0.5
    \begin{bmatrix}
\frac{\sqrt{3}+1}{\sqrt{3}} & \frac{\sqrt{3}}{3}(1+i)\\
\frac{\sqrt{3}}{3}(1-i) & \frac{\sqrt{3}-1}{\sqrt{3}} 
\end{bmatrix}
\\
\rho^{SIC-POVM}_2 = 0.5
    \begin{bmatrix}
\frac{\sqrt{3}-1}{\sqrt{3}} & \frac{\sqrt{3}}{3}(1-i)\\
\frac{\sqrt{3}}{3}(1+i) & \frac{\sqrt{3}+1}{\sqrt{3}} 
\end{bmatrix}
\\
\rho^{SIC-POVM}_3 = 0.5
    \begin{bmatrix}
\frac{\sqrt{3}-1}{\sqrt{3}} & \frac{\sqrt{3}}{3}(-1-i)\\
\frac{\sqrt{3}}{3}(-1+i) & \frac{\sqrt{3}+1}{\sqrt{3}} 
\end{bmatrix}
\\
\rho^{SIC-POVM}_4 = 0.5
    \begin{bmatrix}
\frac{\sqrt{3}-1}{\sqrt{3}} & \frac{\sqrt{3}}{3}(-1+i)\\
\frac{\sqrt{3}}{3}(-1-i) & \frac{\sqrt{3}+1}{\sqrt{3}} 
\end{bmatrix}
\\
\rho^{Bell}_1 = 0.5
    \begin{bmatrix}
1 \\ 0 \\ 0 \\ 1 
\end{bmatrix}
\\
\rho^{Bell}_2 = 0.5
    \begin{bmatrix}
0 \\ 1 \\ 1 \\ 0 
\end{bmatrix}
\\
\rho^{Bell}_3 = 0.5
    \begin{bmatrix}
1 \\ 0 \\ 0 \\ -1 
\end{bmatrix}
\\
\rho^{Bell}_4 = 0.5
    \begin{bmatrix}
0 \\ 1 \\ -1 \\ 0 
\end{bmatrix}
\end{split}
\end{equation}

As said, the SIC-POVM states can be presented as Bloch vectors:

    \begin{equation}
        \begin{split}
            r_1^{Bell} = \frac{1}{\sqrt{3}}
            \begin{bmatrix}
                1 \\ -1 \\ 1
            \end{bmatrix}
            \\
            r_2^{Bell} = \frac{1}{\sqrt{3}}
            \begin{bmatrix}
                1 \\ 1 \\ -1
            \end{bmatrix}
            \\
            r_3^{Bell} = \frac{1}{\sqrt{3}}
            \begin{bmatrix}
                -1 \\ 1 \\ 1
            \end{bmatrix}
\\
            r_4^{Bell} = \frac{1}{\sqrt{3}}
            \begin{bmatrix}
                -1 \\ -1 \\ -1
            \end{bmatrix}
                    \end{split}
    \end{equation}
Note that in order for our construction to work, the indices of the states  cannot be changed.

%\acknowledgments{Special thanks to Etai Flint for Figures~\ref{tensor_slices},\ref{tensor_figure},\ref{tensor_interior}, and ``DoctorBear" on Youtube for the MATLAB Bloch sphere plotting tutorial, used in Figure~\ref{SIC-POVM_figure}.}

\bibliographystyle{apsrev4-1}
\bibliography{ArticleReview}

\begin{thebibliography}{22}%
\makeatletter
\providecommand \@ifxundefined [1]{%
 \@ifx{#1\undefined}
}%
\providecommand \@ifnum [1]{%
 \ifnum #1\expandafter \@firstoftwo
 \else \expandafter \@secondoftwo
 \fi
}%
\providecommand \@ifx [1]{%
 \ifx #1\expandafter \@firstoftwo
 \else \expandafter \@secondoftwo
 \fi
}%
\providecommand \natexlab [1]{#1}%
\providecommand \enquote  [1]{``#1''}%
\providecommand \bibnamefont  [1]{#1}%
\providecommand \bibfnamefont [1]{#1}%
\providecommand \citenamefont [1]{#1}%
\providecommand \href@noop [0]{\@secondoftwo}%
\providecommand \href [0]{\begingroup \@sanitize@url \@href}%
\providecommand \@href[1]{\@@startlink{#1}\@@href}%
\providecommand \@@href[1]{\endgroup#1\@@endlink}%
\providecommand \@sanitize@url [0]{\catcode `\\12\catcode `\$12\catcode
  `\&12\catcode `\#12\catcode `\^12\catcode `\_12\catcode `\%12\relax}%
\providecommand \@@startlink[1]{}%
\providecommand \@@endlink[0]{}%
\providecommand \url  [0]{\begingroup\@sanitize@url \@url }%
\providecommand \@url [1]{\endgroup\@href {#1}{\urlprefix }}%
\providecommand \urlprefix  [0]{URL }%
\providecommand \Eprint [0]{\href }%
\providecommand \doibase [0]{http://dx.doi.org/}%
\providecommand \selectlanguage [0]{\@gobble}%
\providecommand \bibinfo  [0]{\@secondoftwo}%
\providecommand \bibfield  [0]{\@secondoftwo}%
\providecommand \translation [1]{[#1]}%
\providecommand \BibitemOpen [0]{}%
\providecommand \bibitemStop [0]{}%
\providecommand \bibitemNoStop [0]{.\EOS\space}%
\providecommand \EOS [0]{\spacefactor3000\relax}%
\providecommand \BibitemShut  [1]{\csname bibitem#1\endcsname}%
\let\auto@bib@innerbib\@empty
%</preamble>
\bibitem [{\citenamefont {Peres}(1996)}]{peres1996separability}%
  \BibitemOpen
  \bibfield  {author} {\bibinfo {author} {\bibfnamefont {A.}~\bibnamefont
  {Peres}},\ }\href@noop {} {\bibfield  {journal} {\bibinfo  {journal} {Phys.
  Rev. Lett.}\ }\textbf {\bibinfo {volume} {77}},\ \bibinfo {pages} {1413}
  (\bibinfo {year} {1996})}\BibitemShut {NoStop}%
\bibitem [{\citenamefont {Horodecki}\ \emph {et~al.}(1996)\citenamefont
  {Horodecki}, \citenamefont {Horodecki},\ and\ \citenamefont
  {Horodecki}}]{HORODECKI19961}%
  \BibitemOpen
  \bibfield  {author} {\bibinfo {author} {\bibfnamefont {M.}~\bibnamefont
  {Horodecki}}, \bibinfo {author} {\bibfnamefont {P.}~\bibnamefont
  {Horodecki}}, \ and\ \bibinfo {author} {\bibfnamefont {R.}~\bibnamefont
  {Horodecki}},\ }\href {\doibase
  https://doi.org/10.1016/S0375-9601(96)00706-2} {\bibfield  {journal}
  {\bibinfo  {journal} {Physics Letters A}\ }\textbf {\bibinfo {volume}
  {223}},\ \bibinfo {pages} {1} (\bibinfo {year} {1996})}\BibitemShut {NoStop}%
\bibitem [{\citenamefont {Chen}\ and\ \citenamefont
  {Wu}(2002)}]{chen2002matrix}%
  \BibitemOpen
  \bibfield  {author} {\bibinfo {author} {\bibfnamefont {K.}~\bibnamefont
  {Chen}}\ and\ \bibinfo {author} {\bibfnamefont {L.-A.}\ \bibnamefont {Wu}},\
  }\href@noop {} {\bibfield  {journal} {\bibinfo  {journal} {Quantum Inf.
  Comput.}\ }\textbf {\bibinfo {volume} {3}} (\bibinfo {year}
  {2002})}\BibitemShut {NoStop}%
\bibitem [{\citenamefont {Rudolph}(2005)}]{rudolph2005further}%
  \BibitemOpen
  \bibfield  {author} {\bibinfo {author} {\bibfnamefont {O.}~\bibnamefont
  {Rudolph}},\ }\href@noop {} {\bibfield  {journal} {\bibinfo  {journal}
  {Quantum Inf. Process.}\ }\textbf {\bibinfo {volume} {4}},\ \bibinfo {pages}
  {219} (\bibinfo {year} {2005})}\BibitemShut {NoStop}%
\bibitem [{\citenamefont {de~Vicente}(2007)}]{vicente2007separability}%
  \BibitemOpen
  \bibfield  {author} {\bibinfo {author} {\bibfnamefont {J.~I.}\ \bibnamefont
  {de~Vicente}},\ }\href@noop {} {\bibfield  {journal} {\bibinfo  {journal}
  {Quantum Inf. Comput.}\ }\textbf {\bibinfo {volume} {7}},\ \bibinfo {pages}
  {624–638} (\bibinfo {year} {2007})}\BibitemShut {NoStop}%
\bibitem [{\citenamefont {G{\"u}hne}\ and\ \citenamefont
  {T{\'o}th}(2009)}]{guhne2009entanglement}%
  \BibitemOpen
  \bibfield  {author} {\bibinfo {author} {\bibfnamefont {O.}~\bibnamefont
  {G{\"u}hne}}\ and\ \bibinfo {author} {\bibfnamefont {G.}~\bibnamefont
  {T{\'o}th}},\ }\href@noop {} {\bibfield  {journal} {\bibinfo  {journal}
  {Phys. Rep.}\ }\textbf {\bibinfo {volume} {474}},\ \bibinfo {pages} {1}
  (\bibinfo {year} {2009})}\BibitemShut {NoStop}%
\bibitem [{\citenamefont {Peled}\ \emph {et~al.}(2021)\citenamefont {Peled},
  \citenamefont {Te’eni}, \citenamefont {Carmi},\ and\ \citenamefont
  {Cohen}}]{peled2021correlation}%
  \BibitemOpen
  \bibfield  {author} {\bibinfo {author} {\bibfnamefont {B.~Y.}\ \bibnamefont
  {Peled}}, \bibinfo {author} {\bibfnamefont {A.}~\bibnamefont {Te’eni}},
  \bibinfo {author} {\bibfnamefont {A.}~\bibnamefont {Carmi}}, \ and\ \bibinfo
  {author} {\bibfnamefont {E.}~\bibnamefont {Cohen}},\ }\href@noop {}
  {\bibfield  {journal} {\bibinfo  {journal} {Scientific reports}\ }\textbf
  {\bibinfo {volume} {11}},\ \bibinfo {pages} {1} (\bibinfo {year}
  {2021})}\BibitemShut {NoStop}%
\bibitem [{\citenamefont {Carmi}\ and\ \citenamefont
  {Cohen}(2019)}]{carmi2019relativistic}%
  \BibitemOpen
  \bibfield  {author} {\bibinfo {author} {\bibfnamefont {A.}~\bibnamefont
  {Carmi}}\ and\ \bibinfo {author} {\bibfnamefont {E.}~\bibnamefont {Cohen}},\
  }\href@noop {} {\bibfield  {journal} {\bibinfo  {journal} {Sci. Adv.}\
  }\textbf {\bibinfo {volume} {5}},\ \bibinfo {pages} {eaav8370} (\bibinfo
  {year} {2019})}\BibitemShut {NoStop}%
\bibitem [{\citenamefont {Badziag}\ \emph {et~al.}(2008)\citenamefont
  {Badziag}, \citenamefont {Brukner}, \citenamefont {Laskowski}, \citenamefont
  {Paterek},\ and\ \citenamefont {{\.Z}ukowski}}]{badziag2008experimentally}%
  \BibitemOpen
  \bibfield  {author} {\bibinfo {author} {\bibfnamefont {P.}~\bibnamefont
  {Badziag}}, \bibinfo {author} {\bibfnamefont {{\v{C}}.}~\bibnamefont
  {Brukner}}, \bibinfo {author} {\bibfnamefont {W.}~\bibnamefont {Laskowski}},
  \bibinfo {author} {\bibfnamefont {T.}~\bibnamefont {Paterek}}, \ and\
  \bibinfo {author} {\bibfnamefont {M.}~\bibnamefont {{\.Z}ukowski}},\
  }\href@noop {} {\bibfield  {journal} {\bibinfo  {journal} {Physical review
  letters}\ }\textbf {\bibinfo {volume} {100}},\ \bibinfo {pages} {140403}
  (\bibinfo {year} {2008})}\BibitemShut {NoStop}%
\bibitem [{\citenamefont {de~Vicente}\ and\ \citenamefont
  {Huber}(2011)}]{de2011multipartite}%
  \BibitemOpen
  \bibfield  {author} {\bibinfo {author} {\bibfnamefont {J.~I.}\ \bibnamefont
  {de~Vicente}}\ and\ \bibinfo {author} {\bibfnamefont {M.}~\bibnamefont
  {Huber}},\ }\href@noop {} {\bibfield  {journal} {\bibinfo  {journal} {Phys.
  Rev. A}\ }\textbf {\bibinfo {volume} {84}},\ \bibinfo {pages} {062306}
  (\bibinfo {year} {2011})}\BibitemShut {NoStop}%
\bibitem [{\citenamefont {D{\"u}r}\ and\ \citenamefont
  {Cirac}(2000)}]{dur2000classification}%
  \BibitemOpen
  \bibfield  {author} {\bibinfo {author} {\bibfnamefont {W.}~\bibnamefont
  {D{\"u}r}}\ and\ \bibinfo {author} {\bibfnamefont {J.~I.}\ \bibnamefont
  {Cirac}},\ }\href@noop {} {\bibfield  {journal} {\bibinfo  {journal}
  {Physical Review A}\ }\textbf {\bibinfo {volume} {61}},\ \bibinfo {pages}
  {042314} (\bibinfo {year} {2000})}\BibitemShut {NoStop}%
\bibitem [{\citenamefont {Verstraete}\ \emph {et~al.}(2003)\citenamefont
  {Verstraete}, \citenamefont {Dehaene},\ and\ \citenamefont
  {De~Moor}}]{verstraete2003normal}%
  \BibitemOpen
  \bibfield  {author} {\bibinfo {author} {\bibfnamefont {F.}~\bibnamefont
  {Verstraete}}, \bibinfo {author} {\bibfnamefont {J.}~\bibnamefont {Dehaene}},
  \ and\ \bibinfo {author} {\bibfnamefont {B.}~\bibnamefont {De~Moor}},\
  }\href@noop {} {\bibfield  {journal} {\bibinfo  {journal} {Phys. Rev. A}\
  }\textbf {\bibinfo {volume} {68}},\ \bibinfo {pages} {012103} (\bibinfo
  {year} {2003})}\BibitemShut {NoStop}%
\bibitem [{\citenamefont {Leverrier}\ \emph {et~al.}(2013)\citenamefont
  {Leverrier}, \citenamefont {Garc{\'\i}a-Patr{\'o}n}, \citenamefont {Renner},\
  and\ \citenamefont {Cerf}}]{leverrier2013security}%
  \BibitemOpen
  \bibfield  {author} {\bibinfo {author} {\bibfnamefont {A.}~\bibnamefont
  {Leverrier}}, \bibinfo {author} {\bibfnamefont {R.}~\bibnamefont
  {Garc{\'\i}a-Patr{\'o}n}}, \bibinfo {author} {\bibfnamefont {R.}~\bibnamefont
  {Renner}}, \ and\ \bibinfo {author} {\bibfnamefont {N.~J.}\ \bibnamefont
  {Cerf}},\ }\href@noop {} {\bibfield  {journal} {\bibinfo  {journal} {Physical
  review letters}\ }\textbf {\bibinfo {volume} {110}},\ \bibinfo {pages}
  {030502} (\bibinfo {year} {2013})}\BibitemShut {NoStop}%
\bibitem [{\citenamefont {Sheridan}\ \emph {et~al.}(2010)\citenamefont
  {Sheridan}, \citenamefont {Le},\ and\ \citenamefont
  {Scarani}}]{sheridan2010finite}%
  \BibitemOpen
  \bibfield  {author} {\bibinfo {author} {\bibfnamefont {L.}~\bibnamefont
  {Sheridan}}, \bibinfo {author} {\bibfnamefont {T.~P.}\ \bibnamefont {Le}}, \
  and\ \bibinfo {author} {\bibfnamefont {V.}~\bibnamefont {Scarani}},\
  }\href@noop {} {\bibfield  {journal} {\bibinfo  {journal} {New Journal of
  Physics}\ }\textbf {\bibinfo {volume} {12}},\ \bibinfo {pages} {123019}
  (\bibinfo {year} {2010})}\BibitemShut {NoStop}%
\bibitem [{\citenamefont {{\c{C}}orbaci}\ \emph {et~al.}(2016)\citenamefont
  {{\c{C}}orbaci}, \citenamefont {Karaka{\c{s}}},\ and\ \citenamefont
  {Gen{\c{c}}ten}}]{ccorbaci2016construction}%
  \BibitemOpen
  \bibfield  {author} {\bibinfo {author} {\bibfnamefont {S.}~\bibnamefont
  {{\c{C}}orbaci}}, \bibinfo {author} {\bibfnamefont {M.~D.}\ \bibnamefont
  {Karaka{\c{s}}}}, \ and\ \bibinfo {author} {\bibfnamefont {A.}~\bibnamefont
  {Gen{\c{c}}ten}},\ }in\ \href@noop {} {\emph {\bibinfo {booktitle} {Journal
  of Physics: Conference Series}}},\ Vol.\ \bibinfo {volume} {766}\ (\bibinfo
  {organization} {IOP Publishing},\ \bibinfo {year} {2016})\ p.\ \bibinfo
  {pages} {012014}\BibitemShut {NoStop}%
\bibitem [{\citenamefont {Fujii}(2001)}]{fujii2001generalized}%
  \BibitemOpen
  \bibfield  {author} {\bibinfo {author} {\bibfnamefont {K.}~\bibnamefont
  {Fujii}},\ }\href@noop {} {\bibfield  {journal} {\bibinfo  {journal} {arXiv
  preprint quant-ph/0106018}\ } (\bibinfo {year} {2001})}\BibitemShut {NoStop}%
\bibitem [{\citenamefont {Karimipour}\ \emph {et~al.}(2002)\citenamefont
  {Karimipour}, \citenamefont {Bahraminasab},\ and\ \citenamefont
  {Bagherinezhad}}]{karimipour2002quantum}%
  \BibitemOpen
  \bibfield  {author} {\bibinfo {author} {\bibfnamefont {V.}~\bibnamefont
  {Karimipour}}, \bibinfo {author} {\bibfnamefont {A.}~\bibnamefont
  {Bahraminasab}}, \ and\ \bibinfo {author} {\bibfnamefont {S.}~\bibnamefont
  {Bagherinezhad}},\ }\href@noop {} {\bibfield  {journal} {\bibinfo  {journal}
  {Physical Review A}\ }\textbf {\bibinfo {volume} {65}},\ \bibinfo {pages}
  {052331} (\bibinfo {year} {2002})}\BibitemShut {NoStop}%
\bibitem [{\citenamefont {Sarbicki}\ \emph {et~al.}(2020)\citenamefont
  {Sarbicki}, \citenamefont {Scala},\ and\ \citenamefont
  {Chru{\'s}ci{\'n}ski}}]{sarbicki2020family}%
  \BibitemOpen
  \bibfield  {author} {\bibinfo {author} {\bibfnamefont {G.}~\bibnamefont
  {Sarbicki}}, \bibinfo {author} {\bibfnamefont {G.}~\bibnamefont {Scala}}, \
  and\ \bibinfo {author} {\bibfnamefont {D.}~\bibnamefont
  {Chru{\'s}ci{\'n}ski}},\ }\href@noop {} {\bibfield  {journal} {\bibinfo
  {journal} {Physical Review A}\ }\textbf {\bibinfo {volume} {101}},\ \bibinfo
  {pages} {012341} (\bibinfo {year} {2020})}\BibitemShut {NoStop}%
\bibitem [{\citenamefont {Luo}\ and\ \citenamefont
  {Fu}(2010)}]{luo2010geometric}%
  \BibitemOpen
  \bibfield  {author} {\bibinfo {author} {\bibfnamefont {S.}~\bibnamefont
  {Luo}}\ and\ \bibinfo {author} {\bibfnamefont {S.}~\bibnamefont {Fu}},\
  }\href@noop {} {\bibfield  {journal} {\bibinfo  {journal} {Phys. Rev. A}\
  }\textbf {\bibinfo {volume} {82}},\ \bibinfo {pages} {034302} (\bibinfo
  {year} {2010})}\BibitemShut {NoStop}%
\bibitem [{\citenamefont {Rulli}\ and\ \citenamefont
  {Sarandy}(2011)}]{rulli2011global}%
  \BibitemOpen
  \bibfield  {author} {\bibinfo {author} {\bibfnamefont {C.}~\bibnamefont
  {Rulli}}\ and\ \bibinfo {author} {\bibfnamefont {M.}~\bibnamefont
  {Sarandy}},\ }\href@noop {} {\bibfield  {journal} {\bibinfo  {journal}
  {Physical Review A}\ }\textbf {\bibinfo {volume} {84}},\ \bibinfo {pages}
  {042109} (\bibinfo {year} {2011})}\BibitemShut {NoStop}%
\bibitem [{\citenamefont {Xu}(2012)}]{xu2012geometric}%
  \BibitemOpen
  \bibfield  {author} {\bibinfo {author} {\bibfnamefont {J.}~\bibnamefont
  {Xu}},\ }\href@noop {} {\bibfield  {journal} {\bibinfo  {journal} {Journal of
  Physics A: Mathematical and Theoretical}\ }\textbf {\bibinfo {volume} {45}},\
  \bibinfo {pages} {405304} (\bibinfo {year} {2012})}\BibitemShut {NoStop}%
\bibitem [{\citenamefont {Hiai}\ and\ \citenamefont
  {Petz}(2014)}]{hiai2014introduction}%
  \BibitemOpen
  \bibfield  {author} {\bibinfo {author} {\bibfnamefont {F.}~\bibnamefont
  {Hiai}}\ and\ \bibinfo {author} {\bibfnamefont {D.}~\bibnamefont {Petz}},\
  }\href@noop {} {\emph {\bibinfo {title} {Introduction to matrix analysis and
  applications}}}\ (\bibinfo  {publisher} {Springer Science \& Business
  Media},\ \bibinfo {year} {2014})\BibitemShut {NoStop}%
\end{thebibliography}%
\end{document}